\documentclass[12pt,reqno]{amsart}
\usepackage[hypertex]{hyperref}
\usepackage[final]{graphicx}
\usepackage{amsfonts}
\usepackage{pdfsync}
\usepackage{amsmath}
\usepackage{amssymb}
\usepackage{amsthm}

\newcommand{\dd}{\mathrm{d}}

\newtheorem{theorem}{Theorem}[section]
\newtheorem{lemma}[theorem]{Lemma}

\newtheorem{remark}[theorem]{Remark}
\newtheorem*{remarks}{Remarks}

\newcommand{\R}{{\mathord{\mathbb R}}}

\begin{document}
\title[Indirect Coulomb Energy for Two-Dimensional Atoms]{\textbf{Indirect Coulomb Energy for Two-Dimensional Atoms}}

\author[Benguria]{Rafael D. Benguria$^1$}

\author[Tu\v{s}ek]{Mat\v{e}j Tu\v{s}ek$^2$}

\address{$^1$ Departmento de F\'\i sica, P. Universidad Cat\' olica de Chile,}
\email{\href{mailto: rbenguri@fis.puc.cl}{rbenguri@fis.puc.cl}}

\address{$^2$ Departmento de F\'\i sica, P. Universidad Cat\' olica de Chile, }
\email{\href{mailto: mtusek@fis.puc.cl}{mtusek@fis.puc.cl}}


\smallskip
\begin{abstract}
In this manuscript we provide a family of lower bounds on the indirect Coulomb energy 
for atomic and molecular systems in two dimensions in terms of a functional of the single 
particle density with gradient correction terms. 
\end{abstract}

\maketitle

\section{Introduction}

Since the advent of quantum mechanics, the
impossibility of solving exactly problems involving
many particles has been clear. These problems are
of interest in such areas as atomic and molecular
physics, condensed matter physics, and nuclear
physics. It was, therefore, necessary from the early
beginnings to estimate various energy terms of a system 
of electrons as functionals of the single particle density $\rho_{\psi}(x)$, 
rather than as functionals of their wave function $\psi$. 
The first estimates of this type were obtained by 
Thomas and Fermi in 1927 (see \cite{Li81} for a review), and by now they have 
given rise to a whole discipline under the name of {\it Density Functional Theory} (see, e.g., 
\cite{Be13} and references therein). 
In Quantum Mechanics of many particle
systems the main object of interest is the wavefunction $\psi \in \bigwedge^N L^2({\R}^3)$, 
(the antisymmetric tensor product of $L^2({\R}^3)$). 
More explicitly, for a system of $N$ fermions, $\psi(x_1, \dots , x_i,  \dots, \dots x_ j, \dots x_N)= - 
\psi(x_1, \dots , x_j,  \dots, \dots x_ i, \dots x_N)$, in view of Pauli's Exclusion Principle, and
$\int_{\R^N} |\psi|^2 \, dx_1 \dots d x_n=1$.  Here, $x_i \in \R^3$ denote the coordinates of 
the $i$-th particle. 
From the wavefunction $\psi$ one can define the one--particle density (single particle density) as
\begin{equation}
\rho_{\psi}(x) = N \int_{\R^{3(N-1)}} |\psi (x, x _2, \dots, x_N)|^2 \, dx_2 \dots dx_N,
\label{density}
\end{equation}
and from here it follows that $\int_{\R^3} \rho_{\psi}  (x) \, dx = N$, the number of particles, and 
$\rho_{\psi}(x)$ is the density of particles at $x \in \R^3$. Notice that since $\psi$ is antisymmetric, $|\psi|^2$ is symmetric, 
and it is immaterial which variable is set equal to $x$ in (\ref{density}).

In Atomic and Molecular Physics, given that the expectation value 
of the Coulomb attraction of the electrons by the nuclei can be expressed 
in closed form in terms of $\rho_{\psi}(x)$, the interest focuses on estimating 
the expectation value of the kinetic energy of the system of electrons and  
the expectation value of the Coulomb repulsion between the electrons.  
Here, we will be concerned with the latest. The most natural approximation 
to the expectation value of the Coulomb repulsion between the electrons is given by
\begin{equation}
D(\rho_{\psi},\rho_{\psi})=\frac{1}{2} \int \rho_{\psi} (x) \frac{1}{|x-y|} \rho_{\psi}(y) \, \dd x \, \dd y, 
\end{equation}
which is usually called the {\it direct term}. 
The remainder, i.e., the difference between the expectation value of the 
electronic repulsion and $D(\rho_{\psi},\rho_{\psi})$, say $E$, is called the {\it indirect term}. 
In 1930, Dirac \cite{Di30} gave the first approximation to the indirect 
Coulomb energy in terms of the single particle density. Using an argument with  
plane waves, he approximated $E$ by 
\begin{equation}
E \approx -c_D  \int \rho_{\psi}^{4/3} \, dx,
\label{eq:dirac}
\end{equation}
where $c_D=(3/4)(3/\pi)^{1/3} \approx  0.7386$ (see, e.g., \cite{Mo06}, p. 299). Here we use
units in which  the absolute value of the charge of the electron is one.  
The first rigorous lower bound for $E$ was obtained by E.H. Lieb in 1979 \cite{Li79}, using the Hardy--Littlewood Maximal Function \cite{StWe71}. There he found that,
$E \geq -8.52  \int \rho_{\psi}^{4/3} \, dx$.  The constant  $8.52$ was substantially improved by E.H. Lieb and S. Oxford in 1981 \cite{LiOx81}, who proved the bound
\begin{equation}
E  \ge  -C  \int \rho_{\psi}^{4/3} \, dx,
\label{eq:LO}
\end{equation}
with  $C = C_{LO}=1.68$.   In their proof, Lieb and Oxford used Onsager's electrostatic inequality \cite{On39}, and a localization argument. 
The best value for $C$ is unknown, but Lieb and Oxford \cite{LiOx81} 
proved that it is larger or equal than $1.234$. The Lieb--Oxford value was later improved to $1.636$ by 
Chan and Handy, in 1999 \cite{ChHa99}.  
Since  the work of Lieb and Oxford \cite{LiOx81},  there has been a 
special interest in quantum chemistry in constructing corrections to the Lieb--Oxford term involving the 
gradient of the single particle density. This interest arises with the expectation that states with a relatively small kinetic energy have a smaller indirect part (see, e.g., \cite{LePe93,PeBuEr96,VeMeTr09} and references therein). Recently, Benguria, Bley, and Loss obtained an alternative to (\ref{eq:LO}), which has a lower constant (close to $1.45$) to the expense of adding a gradient term (see Theorem 1.1 in  \cite{BeBlLo12}), which we state below in a slightly modified way, 

\begin{theorem}[Benguria, Bley, Loss \cite{BeBlLo12}]\label{BBL}
For any normalized wave function $\psi(x_1, \dots, x_N)$ and any $0 < \epsilon < 1/2$ we have the estimate
\begin{equation}
E(\psi) \ge -  1.4508 \, (1+\epsilon) \int_{\R^3} \rho_{\psi}^{4/3} dx -\frac{3}{2 \epsilon} (\sqrt{\rho_{\psi}}, |p| \sqrt{\rho_{\psi}})
\label{exch}
\end{equation}
where
\begin{equation}
(\sqrt \rho, |p| \sqrt \rho) :=  \int_{\R^3} |\widehat{\sqrt \rho}(k)|^2 |2\pi k| d k = \frac{1}{2\pi^2} \int_{\R^3} \int_{\R^3} \frac{|\sqrt{\rho(x)} - \sqrt{\rho(y)}|^2 }{|x-y|^4} dx dy  \ .
\label{eq:KE}
\end{equation}
Here,  $\widehat f(k)$ denotes the Fourier-transform
$$
\widehat f (k) = \int_{\R^3} e^{-2\pi i k \cdot x} f(x) d x\ .
$$
\end{theorem}

\begin{remarks}

i) For many physical states the contribution of the last two terms in (\ref{exch}) is small compared with the contribution of the first term. See, e.g., the  Appendix in \cite{BeBlLo12}; 

ii) For the second equality in (\ref{eq:KE})  see, e.g., \cite{LiLo01}, Section 7.12, equation (4), p. 184;

iii) It was already noticed by Lieb and Oxford (see the remark after equation (26), p. 261 on \cite{LiOx81}), that somehow for uniform densities the Lieb--Oxford constant should be $1.45$ instead of $1.68$; 

iv) In the same vein, J. P. Perdew \cite{Pe91}, by employing results for a uniform electron gas in its low density limit, showed that in the Lieb--Oxford bound one ought to have $C \ge 1.43$ (see also \cite{LePe93}).

\end{remarks}

After the work of Lieb and Oxford \cite{LiOx81} many people have considered bounds on the indirect Coulomb energy in lower dimensions (in particular see, e.g., \cite{HaSe01} for the one-dimensional case;  \cite{LiSoYn95}, \cite{NaPoSo11}, \cite{RaPiCaPr09}, and \cite{RaSeGo11} for the two-dimensional case, which is important for the study of quantum dots). Recently, Benguria, Gallegos, and Tu\v{s}ek \cite{BeGaTu12} gave an alternative to the Lieb--Solovej--Yngvason bound \cite{LiSoYn95}, with a constant  much closer to the numerical values proposed in \cite{RaSeGo11} (see also the references therein) 
to the expense of adding a gradient term:

\begin{theorem}[Estimate on the indirect Coulomb energy for two dimensional atoms \cite{BeGaTu12}]\label{thm:LO}
Let $\psi\in L^{2}(\R^{2N})$ be normalized to one and symmetric (or antisymmetric) in all its variables. Define
$$\rho_{\psi}(x)=N\int_{\R^{2(N-1)}}|\psi|^{2}(x,x_{2},\ldots,x_{N})~\dd x_{2}\ldots\dd x_{N}.$$
If $\rho_{\psi}\in L^{3/2}(\R^2)$ and $|\nabla\rho_{\psi}^{1/4}|\in L^2(\R^2)$, then, for all $\epsilon>0$, 
\begin{equation}
 E(\psi)\equiv\langle\psi,\sum_{i<j}^{N}|x_{i}-x_{j}|^{-1}\psi\rangle-D(\rho_{\psi},\rho_{\psi})\geq -(1+\epsilon)\beta\int_{\R^{2}}\rho_{\psi}^{3/2} \, \dd x-\frac{4}{\beta\epsilon}\int_{\R^{2}}|\nabla\rho_{\psi}^{1/4}|^{2} \, \dd x
\label{eq:ind_en_est}
\end{equation}
with
\begin{equation}
\beta=\left(\frac{4}{3}\right)^{3/2}\sqrt{5\pi-1}\simeq 5.9045. 
\label{beta}
\end{equation}
\end{theorem}

\begin{remarks}

\bigskip
\noindent
i) The  constant $\beta \simeq 5.9045$ in (\ref{eq:ind_en_est}) is substantially lower than the  constant $C_{LSY} \simeq 481.27$ found in \cite{LiSoYn95} (see equation (5.24) of lemma 5.3 in  \cite{LiSoYn95}).

\bigskip
\noindent
ii) Moreover, the constant $\beta$  is close to the numerical values (i.e., $\simeq 1.95$)  of \cite{RaPiCaPr09} (and references therein), but is not sharp.

\end{remarks} 

In the literature there are, so far,  three approaches to prove lower bounds on the exchange energy, namely: 

\bigskip
\noindent
i)  The approach introduced  by E.H. Lieb in 1979 \cite{Li79}, which uses as the main tool the Hardy--Littlewood Maximal Function \cite{StWe71}. This method was used in the first bound  of Lieb \cite{Li79}. Later it was used in \cite{LiSoYn95} to obtain a lower bound on the exchange energy of two--dimensional Coulomb systems. It has the advantage that it may be applied in a wide class of problems, but it does not yield sharp constants. 

\bigskip
\noindent
 ii) The use of Onsager's electrostatic inequality \cite{On39} together with localization techniques, introduced by Lieb and Oxford \cite{LiOx81}. This method yields very sharp constants. It was used recently in \cite{BeBlLo12} to get a new type of bounds including gradient terms (for three dimensional Coulomb systems). 
 In some sense the constant $1.4508$ recently obtained in \cite{BeBlLo12} is the best possible (see the comments after Theorem \ref{BBL}). The only disadvantage of this approach  is that it depends  on the use of Onsager's electrostatic inequality (which in turn relies  on the fact that the Coulomb potential is the fundamental solution of the Laplacian). Because of this, it cannot be used in the case of two--dimensional atoms, because $1/|x|$ is not the fundamental solution of the two--dimensional Laplacian.  
 
 \bigskip
 \noindent
 iii) The use of the stability of matter of an auxiliary many particle system. This idea was used by Lieb and Thirring \cite{LiTh75}  to obtain lower bounds on the kinetic energy of a systems of electrons in terms of the single particle density. In connection with the problem of getting lower bounds on the exchange energy it was used for the first time in \cite{BeGaTu12}, to get a lower bound on the exchange energy of two--dimensional Coulomb systems including gradient terms. This method provides very good, although not sharp, constants.

\bigskip
\bigskip

As we mentioned above, during the last twenty years there has been a 
special interest in quantum chemistry in constructing corrections to the Lieb--Oxford term involving the 
gradients  of the single particle density. This interest arises with the expectation that states with a relatively small kinetic energy have a smaller indirect part (see, e.g., \cite{LePe93,PeBuEr96,VeMeTr09} and references therein). While the form of  leading term (i.e., the dependence as an integral of  $\rho^{4/3}$ in three dimensions or as an  integral of $\rho^{3/2}$ in two dimensions)  is dictated by Dirac's argument (using plane waves), there is no such a clear argument, nor a common agreement concerning the structure of the gradient corrections. The reason we introduced the particular gradient term,  $\int_{\R^{2}}|\nabla\rho_{\psi}^{1/4}|^{2} \, \dd x$ in our earlier work \cite{BeGaTu12}, was basically due to the fact that we already knew the stability of matter arguments for the auxiliary system. However, there is a whole one parameter family of such gradient terms that can be dealt in the same manner. In this manuscript we obtain lower bounds including as gradient terms this one--parameter family. One interesting feature of our bounds is that the constant $\beta$ in front of the leading term remains the same (i.e., its  value is independent of the parameter that labels the different possible gradient terms), while the constant in front  of the gradient term is parameter dependent. 

\bigskip
Our main result is the following theorem.

\begin{theorem}[Estimate on the indirect Coulomb energy for two dimensional atoms]\label{thm:BeTu}
Let $1<\gamma<3$, and 
$\alpha = (3-\gamma)/(2 \gamma)$. 
Assume, $\rho_{\psi}\in L^{3/2}(\R^2)$ 
and $|\nabla \rho_{\psi}^{\alpha}| \in L^{\gamma}(\R^2)$. 
Let $C(p) = 2^{1-p/2}$, for $0<p \le 2$ while $C(p)=1$, for $p \ge 2$. Then, for all  $\epsilon >0$ we have, 
\begin{equation}
 E(\psi)\equiv\langle\psi,\sum_{i<j}^{N}|x_{i}-x_{j}|^{-1}\psi\rangle-D(\rho_{\psi},\rho_{\psi})
 \geq 
 -\tilde{b}^2 \int_{\R^{2}}\rho_{\psi}^{3/2} \, \dd x  - \tilde{a}^2 \int_{\R^{2}}|\nabla\rho_{\psi}^{\alpha}|^{\gamma} \, \dd x.
\label{MainLowerBound}
\end{equation}
Here, 
\begin{equation}
\tilde{b}^{2}=\left(\frac{4}{3}\right)^{3/2}{\sqrt{5\pi-1}} \, (1+\epsilon) = \beta \, (1+\epsilon)
\label{constant1}
\end{equation}
where $\beta$ is the same constant that appears in (\ref{beta}). 
Also, 
\begin{equation}
 \tilde{a}^{2}=\frac{2^\gamma C(\gamma)}{3-\gamma}\left(\frac{1}{\beta \, \epsilon} \,\frac{\gamma-1}{3-\gamma}\, C\bigg(\frac{\gamma}{\gamma-1}\bigg)\right)^{\gamma-1}.
 \label{constant2}
\end{equation}
In particular, we have (with a fixed $\epsilon$)
$$
\tilde{a}^{2}|_{\gamma\to 1+}=\sqrt{2}.
$$
\end{theorem}

\begin{remarks}

\bigskip
\noindent
i) Our previous Theorem \ref{thm:LO} is a particular case of Theorem \ref{thm:BeTu}, for the value $\gamma=2$, $\alpha=1/4$. 

\bigskip
\noindent
ii) Notice that $\tilde{b}^2$ is independent of $\gamma$, and it is therefore the same as in \cite{BeGaTu12}. 

\bigskip
\noindent
iii)The constant in front of the gradient term depends on the power $\gamma$ and, of course, on $\epsilon$. However, as $\gamma \to 1+$, this constant converges to $\sqrt{2}$ independently of the value of $\epsilon$. 

\end{remarks}

In the rest of the manuscript we give a sketch of the proof of this theorem, which follows closely  the proof of the particular result \ref{thm:LO} in \cite{BeGaTu12}.

\section{Auxiliary lemmas}

First we need a standard convexity result. 

\begin{lemma}\label{lem:norm_comp}
 Let $x,y\in\R$, and $p>0$. Then
 $$|x|^p+|y|^p\leq C(p)|x+iy|^p,$$
where $C(p)=2^{1-p/2}$ for $0<p\leq 2$, and $C(p)=1$ for $p\geq 2$. The constant $C(p)$ is sharp.
\end{lemma}
\begin{proof}
 If $p\geq 2$, the assertion follows, e.g.,  from the fact that $l^{p}$-norm is decreasing in $p$. On the other hand, for $0<p<2$, the assertion follows from the concavity of the mapping $t \to t^{1/n}$ for $t>0$ and $n>1$. 
 \end{proof}

The next lemma is a generalization of the analogous result introduced in \cite{BeLoSi07} and used in  the proof of Theorem  \ref{thm:LO} above (see \cite{BeGaTu12}). This lemma is later needed to prove a Coulomb Uncertainty Principle. 

\begin{lemma} \label{lem:uncert_princ}
 Let $D_{R}$ stands for the disk of radius $R$ and origin $(0,0)$. Moreover let  $u=u(|x|)$ be a smooth function such that $u(R)=0$ and $1<\gamma<3$. Then the following uncertainty principle holds
\begin{equation*}
\begin{split}
 &\left|\int_{D_{R}}\big[ 2u(|x|)+|x|u'(|x|)\big]f(x)^{1/\alpha}\right|\leq\\
 &\leq\frac{1}{\alpha}\left(C(\gamma)\int_{D_{R}}|\nabla f(x)|^{\gamma}\, \dd x\right)^{1/\gamma}\left(C(\delta)\int_{D_{R}}|x|^{\delta}|u(|x|)|^{\delta}|f(x)|^{3/(2\alpha)}\, \dd x\right)^{1/\delta},
\end{split}
\end{equation*}
where 
\begin{equation} \label{eq:coeff}
\frac{1}{\alpha}=\frac{2\gamma}{3-\gamma},\qquad \frac{1}{\gamma}+\frac{1}{\delta}=1.
\end{equation}
\end{lemma}
\begin{proof}
 Set $g_{j}(x)=u(|x|)x_{j}$. Then we have,
\begin{equation*}
 \begin{split}
  &\int_{D_{R}}[2u(|x|)+|x|u'(|x|)]f(x)^{1/\alpha}\, \dd x=\sum_{j=1}^{2}\int_{D_{R}}[\partial_{j}g_{j}(x)]f(x)^{1/\alpha}\, \dd x= \\
  &=\sum_{j}\int_{D_{R}}f(x)\partial_{j}[g_{j}(x)f(x)^{1/\alpha-1}]\, \dd x-\left(\frac{1}{\alpha}-1\right)\sum_{j}\int_{D_{R}}f(x)^{1/\alpha-1}g_{j}(x)\partial_{j}f(x)\, \dd x=\\
  &=-\frac{1}{\alpha}\int_{D_{R}}\langle \nabla f(x),\, x\rangle u(|x|)f(x)^{1/\alpha-1}\, \dd x.
 \end{split}
\end{equation*}
In the last equality we integrated by parts and made use of the fact that $u$ vanishes on the boundary $\partial D_{R}$. Next, the H\"{o}lder inequality implies 
\begin{equation*}
\begin{split}
 &\left|\int_{D_{R}}\big[ 2u(|x|)+|x|u'(|x|)\big]f(x)^{1/\alpha}\right|\leq \\
 &\frac{1}{\alpha}\left(\int_{D_{R}}\sum_{j=1}^{2}|\partial_{j}f(x)|^{\gamma}\, \dd x\right)^{1/\gamma}\left(\int_{D_{R}}\sum_{j=1}^{2}|x_{j}|^{\delta} |u(|x|)|^{\delta}|f(x)|^{(1/\alpha-1)\delta}\, \dd x\right)^{1/\delta}.
\end{split}
\end{equation*}
The rest follows from Lemma $\ref{lem:norm_comp}$.
\end{proof}

\section{A stability result for an auxiliary two-dimensional molecular system}

Here we follow the method introduced in \cite{BeGaTu12}. That is,  in order to prove our  Lieb--Oxford type bound (with gradient corrections)  in two dimensions we use  a stability of matter type  result on an auxiliary molecular system. This molecular system is an extension of the one studied in \cite{BeGaTu12}, which was adapted from the similar result in three dimensions discussed  in \cite{BeLoSi07} (this last one  corresponds to the zero mass limit of the model introduced in \cite{En87,EnDr87,EnDr88}). We begin with a typical Coulomb Uncertainty Principle which uses the kinetic energy of the electrons in a ball to bound the Coulomb singularities.  

\begin{theorem}
For every smooth non-negative function $\rho$ on the closed disk $D_{R} \subset \R^2$, and for any  $a,b>0$ we have
\begin{equation*}
 ab \, \alpha\left|\int_{D_{R}}\left(\frac{1}{|x|}-\frac{2}{R}\right)\rho(x)\, \dd x\right|\leq \frac{a^\gamma C(\gamma)}{\gamma}\,\int_{D_{R}}|\nabla\rho(x)^{\alpha}|^{\gamma}\, \dd x+\frac{b^\delta C(\delta)}{\delta}\,\int_{D_{R}}\rho^{3/2}\,\dd x,
\end{equation*}
where $1<\gamma< 3$, and $\alpha$ and $\delta$ are as in (\ref{eq:coeff}).
\end{theorem}
\begin{proof}
 In Lemma \ref{lem:uncert_princ} we set $u(r)=1/r-1/R$ and $f=\rho^{\alpha}$. The assertion of the theorem then follows from Young inequality with coefficients $\gamma$ and $\delta$.
\end{proof}

And now we introduce the auxiliary molecular system through the ``energy functional''
\begin{equation} 
\xi(\rho)= \tilde{a}^2 \int_{\R^2} |\nabla \rho^{\alpha}|^\gamma \, \dd x + \tilde{b}^2 \int_{\R^2} \rho^{3/2} \, \dd x - \int_{\R^2} V(x) \rho (x) \, \dd x +D(\rho,\rho) +U,
\label{EnergyFunctional}
\end{equation}
where
$$V(x) = \sum_{i=1}^K \frac{z}{|x-R_i|},\quad D(\rho,\rho) = \frac{1}{2} \int_{\R^2 \times \R^2} \rho(x) \frac{1}{|x-y|} \rho(y) \, \dd x\, \dd y,\quad U = \sum_{1 \le i < j \le K} \frac{z^2}{|R_i-R_j|}$$
with $z>0$ and $R_{i}\in\R^{2}$. As above we assume $1<\gamma<3$, and $\alpha=(3-\gamma)/(2\gamma)$. The choice of $\alpha$ (in terms of $\gamma$) is made in such a way that the first two terms in (\ref{EnergyFunctional}) scale as one over a length. Indeed, let us denote 
$$
K(\rho) \equiv \tilde{a}^2 \int_{\R^2} |\nabla \rho^{\alpha}|^\gamma \, \dd x + \tilde{b}^2 \int_{\R^2} \rho^{3/2} \, \dd x. 
$$
Given any trial function $\rho \in L^1(\R^2)$ and setting  $\rho_{\lambda} (x)  = \lambda^{2} \rho(\lambda x)$ (thus preserving the $L^1$ norm), it is simple to see that with our choice of $\alpha$ we have $K(\rho_{\lambda})= \lambda K(\rho)$. 

If we now introduce constants $a,b_{1},b_{2}>0$ so that
\begin{align}
 &\tilde{a}^{2}=\frac{a^\gamma C(\gamma)}{2\alpha\, \gamma} \label{eq:a_def}\\
 &\tilde{b}^{2}=\frac{b_{2}^{\delta} C(\delta)}{2\alpha \, \delta}+b_{1}^{2}\nonumber
\end{align}
(again with $\delta$ given by (\ref{eq:coeff})), we may use the proof of \cite[Lemma 2.5]{BeGaTu12} step by step. In particular,
\begin{equation*}
 \xi(\rho)\geq b_{1}^{2}\int_{\R^{2}}\rho^{3/2}\dd x-\int_{\R^{2}}V\rho~\dd x+ab_{2}\sum_{j=1}^{K}\int_{B_{j}}\left(\frac{1}{2|x-R_{j}|}-\frac{1}{D_{j}}\right)\rho(x)\dd x+
 D(\rho,\rho)+U,
\end{equation*}
where 
$$
D_j = \frac{1}{2} \min \{|R_k-R_j| \bigm| k \neq j \},$$
and $B_j$ is a disk with center $R_j$ and of radius $D_j$.

Thus as in \cite[Lemma 2.5]{BeGaTu12} we have that, for 
\begin{equation}\label{eq:z_cond}
z\leq ab_{2}/2,
\end{equation}
it holds
\begin{equation}\label{eq:main_est}
 \xi(\rho)\geq \sum_{j=1}^{K}\frac{1}{D_{j}}\left[\frac{z^{2}}{8}-\frac{4}{27b_{1}^{4}}\left(2z^{3}(\pi-1)+\pi a^{3}b_{2}^{3}\right)\right].
\end{equation}
Consequently we arrive at the following theorem.
\begin{theorem}\label{theo:stability}
 For all non-negative functions $\rho$ such that $\rho\in L^{3/2}(\R^{2})$ and $|\nabla\rho^\alpha|\in L^\gamma(\R^{2})$, we have that
\begin{equation}
\xi(\rho) \geq 0,
\label{EFgreaterthanzero}
\end{equation}
provided that 
\begin{equation}
z\leq \max_{\sigma\in(0,1)}h(\sigma)
\label{maxcharge}
\end{equation}
\begin{equation}
 h(\sigma)=\min\left\{\frac{a}{2}\left(\tilde{b}^{2}\, \frac{3-\gamma}{\gamma-1}\, C\bigg(\frac{\gamma}{\gamma-1}\bigg)^{-1} (1-\sigma)\right)^{(\gamma-1)/\gamma},\, \frac{27}{64}\,\frac{\tilde{b^{4}}}{5\pi-1}\sigma^2\right\},
\label{EquationForH}
\end{equation}
with $a$ given by (\ref{eq:a_def}). 
\end{theorem}
In order to arrive at (\ref{EquationForH}) we set $b_{2}$ in (\ref{eq:main_est}) to be the smallest possible under the condition (\ref{eq:z_cond}), i.e., $b_{2}=2z/a$, and we introduced $\sigma=b_{1}^{2}/\tilde{b}^{2}$.

\section{Proof of Theorem \ref{thm:BeTu}}

In this Section we give the proof of the main result of this paper, namely Theorem \ref{thm:BeTu}. We use an idea introduced by Lieb and Thirring in 1975 in their proof of the stability of matter \cite{LiTh75} (see also the review article \cite{Li76} and the recent monograph \cite{LiSe09}). This idea was first used in this context in \cite{BeGaTu12}. 

\begin{proof}[Proof of Theorem \ref{thm:BeTu}]
Consider the inequality (\ref{EFgreaterthanzero}), with $K=N$ (where $N$ is the number of electrons in our original system), $z=1$ (i.e., the charge of the electrons), and $R_i=x_i$ (for all $i=1, \dots, N$). 
With this choice, according to (\ref{maxcharge}), the inequality  (\ref{EFgreaterthanzero}) is valid as long  as $\tilde{a}$ and $\tilde{b}$ (that are now free parameters) satisfy the constraint, 
\begin{equation}
1 \le \max_{\sigma\in(0,1)}h(\sigma)
\label{eq:3.1}
\end{equation}
with $\sigma_{0}$ (which maximizes $h(\sigma)$) such that $h(\sigma_{0})=1$. 
Let us introduce $\epsilon>0$ and set $\sigma_{0}=1/(1+\epsilon)$. Then the smallest $\tilde{b}$ such that the assumptions of Theorem \ref{theo:stability} may be in principle fulfilled reads
\begin{equation}
\tilde{b}^{2}=\left(\frac{4}{3}\right)^{3/2} {\sqrt{5\pi-1}} \, (1+\epsilon).
\label{tildeb}
\end{equation}
Hence $a$ has to be chosen large enough, namely such that
$$
1=\frac{a}{2}\left(\tilde{b}^{2}\, \frac{3-\gamma}{\gamma-1}\, C\bigg(\frac{\gamma}{\gamma-1}\bigg)^{-1} \frac{\epsilon}{1+\epsilon}\right)^{(\gamma-1)/\gamma},$$
which due to (\ref{eq:a_def}) implies
\begin{equation}
 \tilde{a}^{2}=\frac{2^\gamma C(\gamma)}{3-\gamma}\left(\left(\frac{3}{4}\right)^{3/2}(5\pi-1)^{-1/2}\, \frac{1}{\epsilon}\,\frac{\gamma-1}{3-\gamma}\, C\bigg(\frac{\gamma}{\gamma-1}\bigg)\right)^{\gamma-1}.
\label{tildea}
\end{equation}
Since	
$$\lim_{\gamma \to 1+}C(\gamma)=\sqrt{2},\quad \lim_{\gamma \to 1+}\left(\frac{\gamma-1}{3-\gamma}\, C\bigg(\frac{\gamma}{\gamma-1}\bigg)\right)^{\gamma-1}=1,$$
we have (with a fixed $\epsilon$)
$$
\tilde{a}^{2}|_{\gamma \to 1+}=\sqrt{2}.
$$

Then take any normalized wavefunction $\psi(x_1,x_2, \dots, x_N)$, and multiply (\ref{EFgreaterthanzero})  by $| \psi(x_1, \dots, x_N)|^2$ and integrate  over all the electronic configurations, i.e., on $\R^{2N}$. Moreover, take $\rho=\rho_{\psi}(x)$. We get at once, 
\begin{equation}
 E(\psi)\equiv\langle\psi,\sum_{i<j}^{N}|x_{i}-x_{j}|^{-1}\psi\rangle-D(\rho_{\psi},\rho_{\psi})
 \geq 
 - \tilde{a}^2 \int_{\R^2} |\nabla \rho^{\alpha}|^\gamma \, \dd x - \tilde{b}^2 \int_{\R^2} \rho^{3/2} \, \dd x
\label{eq:3.4}
\end{equation}
provided $\tilde{a}$ and $\tilde{b}$ satisfy
 (\ref{tildea}) and (\ref{tildeb}), respectively.
\end{proof}

\begin{remark}
In general the two integral terms in (\ref{MainLowerBound}) are not comparable. If one takes a very rugged $\rho$, normalized to $N$, the gradient term may be very large while the other term can remain small. However, if one takes a smooth $\rho$, the gradient term can be very small as we illustrate in the example below. 
Let us denote
$$
L(\rho)=\int_{\R^2} \rho(x)^{3/2} \, \dd x
$$
and 
$$
G(\rho)=\int_{\R^2} (|\nabla \rho(x)^{\alpha}|)^{\gamma} \, \dd x.
$$
with $\alpha=(3-\gamma)/(2 \gamma)$. 
We will evaluate them for the normal distribution
 $$
 \rho(|x|)=C\mathrm{e}^{-A|x|^{2}}
 $$
where $C,\, A>0$. Some straightforward integration yields
 $$
 L=C^{3/2}\frac{2\pi}{3A},
 $$
while, 
$$
G=C^{\alpha \, \gamma} \pi 2^{\gamma} (A \alpha)^{(\gamma/2)-1} \Gamma\left(1+\frac{\gamma}{2}\right) \gamma^{-(\gamma/2)-1}. 
 $$ 
With $C=NA/\pi$,
$$
\int_{\R^{2}}\rho(|x|) \, \dd x=N,
$$
and we have
$$
\frac{G}{L}= {3} \left(\frac{\sqrt{2}}{\gamma} \right)^{\gamma} \left(\frac{\pi}{N}\right)^{\gamma/2} \Gamma\left(1+\frac{\gamma}{2}\right) (3-\gamma)^{(\gamma/2)-1},
$$
i.e., in the ``large number of particles''  limit, the $G$ term becomes negligible, for all $1< \gamma< 3$.
\end{remark}

\bigskip

\section*{Acknowledgments}

It is a pleasure to dedicate this manuscript  to Elliott Lieb on his eightieth birthday. The scientific achievements of  Elliott Lieb have inspired generations of Mathematical Physicists. 
This work has been supported by the Iniciativa Cient'fica Milenio, ICM (CHILE) project P07--027-F.   The work of RB has also been supported by FONDECYT (Chile) Project 1100679. The work of MT has also been partially supported by the grant 201/09/0811 of the Czech Science Foundation. 

\bibliographystyle{srt}

\end{document}